\documentclass[10pt,conf]{IEEEtran}

\usepackage{graphicx}
\usepackage{amsmath}
\usepackage[ruled,vlined,linesnumbered]{algorithm2e}
\usepackage{subfigure}
\usepackage{mathrsfs}
\usepackage{amsthm}
\usepackage{amssymb}
\usepackage{color}
\usepackage{hyperref}
\usepackage[columnwise,displaymath,mathlines]{lineno}
\usepackage{booktabs}

\graphicspath{ {images/} }

\newtheorem{theorem}{Theorem}
\newtheorem{definition}{Definition}
\newtheorem{lemma}{Lemma}

\begin{document}
%\title{DCG: a Distributed Conjugate Gradient Algorithm for Efficiently Solving Linear Equations}
\title{DCG: Distributed Conjugate Gradient for Efficient Linear Equations Solving}

\author{Haodi~Ping,  Yongcai~Wang, 
        and~Deying~Li
\IEEEcompsocitemizethanks{\IEEEcompsocthanksitem The authors are with School of Information, Renmin University of China, Beijing, P.R.China, 100872.\protect\\}
\thanks{E-mail: \{haodi.ping, ycw, deyingli\}@ruc.edu.cn}% <-this % stops a space
%\thanks{This work was supported in part by the National Natural Science Foundation of China Grant No. 61972404, 12071478, 61732006.}
}

\IEEEtitleabstractindextext{%
\begin{abstract}
%Distributed algorithms to solve linear equations in multi-agent networks have attracted great research attention and many iteration-based distributed algorithms have been developed. The convergence speed is a key factor of distributed algorithms, and it is shown dependent on the spectral radius of the iteration matrix. However, the iteration matrix is determined by the network structure and is hardly pre-tuned, making the iterative-based distributed algorithms may converge very slowly when the spectral radius is close to 1. In contrast, in centralized optimization, the Conjugate Gradient (CG) is a widely adopted idea to speed up the convergence of the centralized solvers, which can guarantee convergence in fixed steps.  In this paper, we propose a general distributed implementation of CG (DCG). DCG only needs local communication and local computation, while inheriting the characteristic of fast convergence. DCG guarantees to converge in $4Hn$ rounds, where $H$ is the maximum hop number of the network and $n$ is the number of nodes. We present the applications of DCG in solving the least square problem and network localization problem. The results show the convergence speed of DCG is three orders of magnitude faster than the widely used Richardson iteration method.
Distributed algorithms to solve linear equations in multi-agent networks have attracted great research attention and many iteration-based distributed algorithms have been developed. The convergence speed is a key factor to be considered for distributed algorithms, and it is shown dependent on the spectral radius of the iteration matrix. However, the iteration matrix is determined by the network structure and is hardly pre-tuned, making the iterative-based distributed algorithms may converge very slowly when the spectral radius is close to 1. In contrast, in centralized optimization, the Conjugate Gradient (CG) is a widely adopted idea to speed up the convergence of the centralized solvers, which can guarantee convergence in fixed steps.  In this paper, we propose a general distributed implementation of CG, called DCG. DCG only needs local communication and local computation, while inheriting the characteristic of fast convergence. DCG guarantees to converge in $4Hn$ rounds, where $H$ is the maximum hop number of the network and $n$ is the number of nodes. We present the applications of DCG in solving the least square problem and network localization problem. The results show the convergence speed of DCG is three orders of magnitude faster than the widely used Richardson iteration method.
\end{abstract}

}

\maketitle

\IEEEdisplaynontitleabstractindextext

% \fi
\IEEEpeerreviewmaketitle

\section{Introduction}

%\begin{figurehere}
%	\begin{algorithm}[H]
%		\caption{multiobjective DE}
%		initialize population $P = \left \{ X_{1}, ... , X_{N} \right \} $\;
%		( \emph{Evolutionary loop}){$g := 1$ to $G_{max}$}
%		{
%			Do things \;
%			Trim the population to size $N$ using nondominated sorting and diversity estimation \;
%		}
%	\end{algorithm}
%\end{figurehere}
%	
In many multi-agent applications, the underlying problem can be reduced to solving a system of linear equations\cite{mou2015distributed}. 
%There exist many centralized methods to solve linear systems directly, e.g., Gaussian elimination,  triangular decomposition, Gauss Jordan elimination, et al. 
Because the autonomous networked agents are usually discretely deployed, each agent only has the access to communicate with direct neighbors. Moreover, in certain scenarios, each agent only desires its own state. Such characteristics consequently give rise to distributed solvers for linear systems.  Different from the centralized solvers, distributed solvers usually adopt iterative manners. The key feature of iterative approaches is \emph{linear iterations}, where each agent receives states of direct neighbors; updates and sends its own state. The barycentric linear localization algorithm\cite{ECHO}, is a typical iteration-based method.  

The convergence speed is a crucial factor of the iteration-based distributed algorithms, which determines whether the distributed algorithm can be used when the application requires a fast response. %, such as quick localization convergence when the network topology changes. 		
%However, in practical multi-agent applications, only a finite number of iterations can be performed. Thus, the unavoidable convergence error is extremely affected by the convergence rate. 
However, the convergence rate of many iterative methods, e.g., Jacobi iteration, Gauss-Seidel iteration, and Richardson iteration\cite{Wolfgang2016}, are characterized by the spectral radius of the iteration matrix. However, the spectral radius is determined by the network topology and is difficult to be pre-tuned. Thus the convergence speed is highly uncertain and may be very slow in some network states.

\begin{figure}
	\centering
	\subfigure[The Richardson Iteration]{
		\begin{minipage}[T]{0.46\linewidth}
			\centering
			\includegraphics[width=1\textwidth]{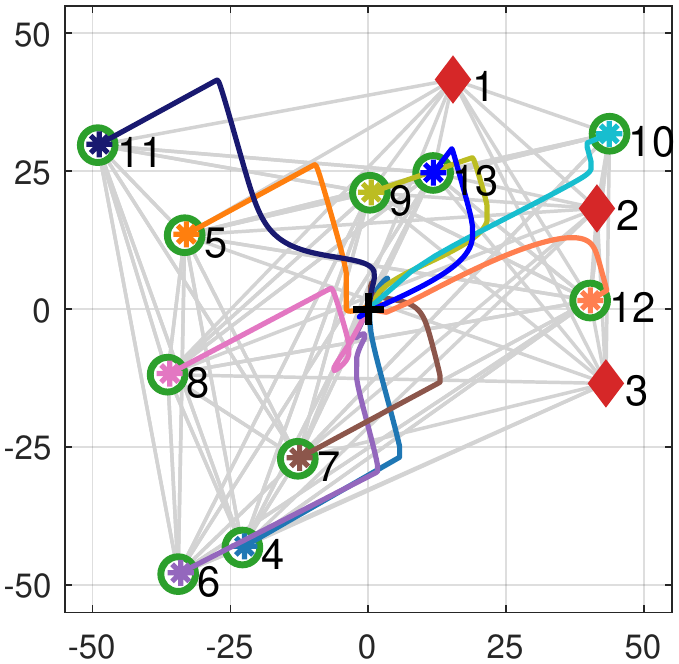}
			\label{fig:trace_richardson_2d}
		\end{minipage}
	}
	\vspace{-0.2cm}
	\subfigure[DCG]{
		\begin{minipage}[T]{0.46\linewidth}
			\centering
			\includegraphics[width=1\textwidth]{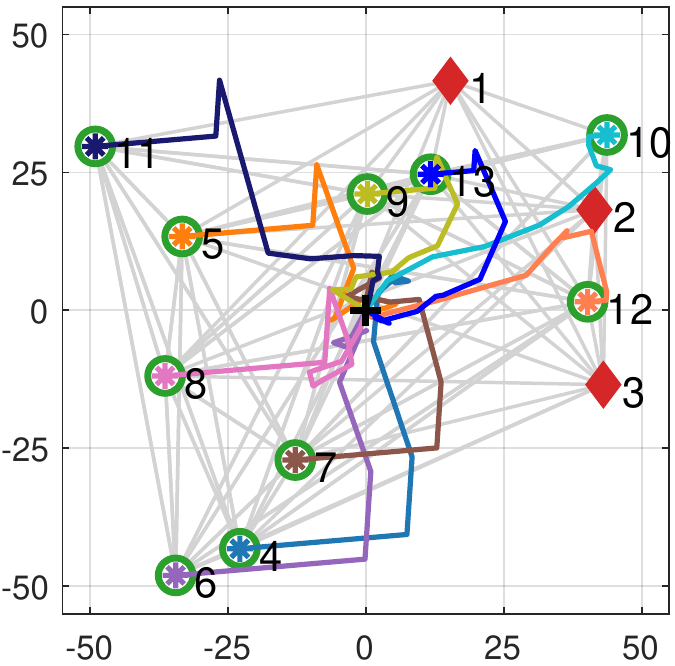}
			\label{fig:trace_dcg_2d}
		\end{minipage}
	}
	\caption{The convergence trails in $\mathbb{R}^2$. The marker `+' represents the start a trail. The marker `$*$' and `$\circ$' represent the converged estimation $\mathbf{\hat{x}}_i$ and the ground truth  $\mathbf{x}_i$, respectively. }
	\label{fig:show_iter_2d}
\end{figure}

\begin{figure}		
	\centering	
	\includegraphics[width=.99\linewidth]{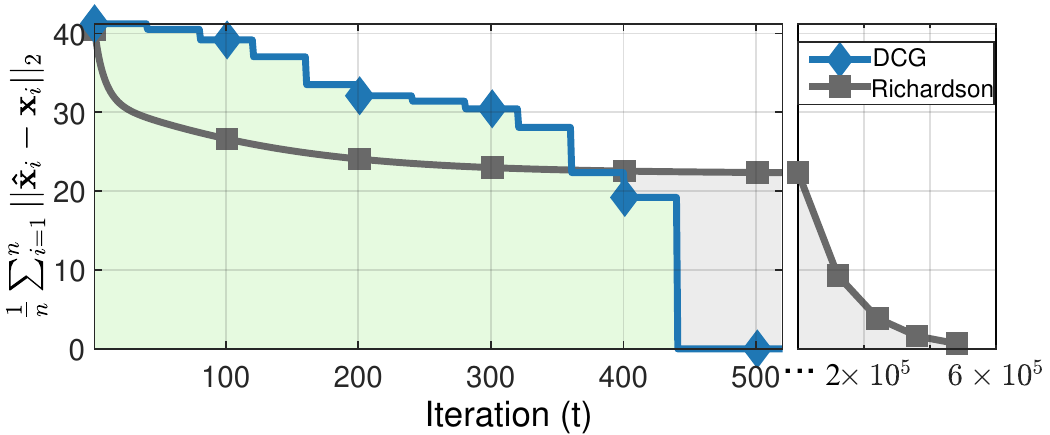}
	\caption{The mean square error w.r.t. iteration rounds.}
	\label{fig:effectiveness_dcg_2d}
\end{figure}

To guarantee fast convergence of distributed algorithms in solving linear equations, this study explores  the idea of  Conjugate Gradient (CG)\cite{aggarwal2020linear}. CG has two desired properties for solving the linear systems: 1) CG converges to the exact solution after a finite number of iterations, which is not larger than the size of the system	matrix; 2)  CG is suited for solving linear systems with large and sparse system matrices.		
%The most attractive  feature of CG is  that it converge within fixed round of iterations. 
But CG is essentially a centralized solver. In pursuit of distributed implementation, we design an efficient protocol to synchronize the necessary vectors to update the residual. Our distributed CG, i.e., DCG remarkably speeds up convergence by paying limited neighborhood communication costs. 

The rest of the paper is organized as follows. We formulate the problem and present related algorithms in Section~\ref{sect:preliminaries}. DCG is proposed in Section~\ref{sect:dcg}.  Applications of DCG are discussed in Section~\ref{sect:applications}. DCG is evaluated in Section~\ref{sect:evaluation}. The paper is concluded with further discussions in Section~\ref{sect:conclusion}. 

\noindent \textbf{Notations:} Throughout the paper, let $\mathbf{AX=b}$ denote a system of linear equations, where $\mathbf{A}\in\mathbb{R}^{n\times n}$ is the \emph{coefficient matrix}, $\mathbf{b}\in\mathbb{R}^{n\times d}$ is the \emph{right-hand side} vector, and  $\mathbf{X}\in\mathbb{R}^{n\times d}$ is the vector of unknowns.  $n$ is the number of variables and $d$ is the spatial dimension of a vector, $d\in\{2,3\}$. The vector $\mathbf{A}_{i,:} = [A_{i1} \cdots A_{in}]$ denotes the $i$th row of $\mathbf{A}$, $i\in\{1\cdots n\}$.

\section{Preliminaries}
\label{sect:preliminaries}
\subsection{Problem Formulation}

Consider a network of $n$ agents $\mathcal{V} = \{v_1,\cdots, v_n\}$,  each agent $v_i$ is capable of communicating with the agents within its reception range $R$. Let $\mathcal{E}$ be the set of edges and $(i,j)\in \mathcal{E}$ if the distance between $v_i$ and $v_j$ is not larger than $R$. Then the multi-agent network can be represented as a graph $\mathcal{G=(V,E)}$. The neighbors of $v_i$ is denoted by $\mathcal{N}_i$, $j\in\mathcal{N}_i$ if $(i,j)\in\mathcal{E}$.

\noindent \textbf{Problem:} Assume that $\mathbf{A}$ is non-singular. Let $\mathbf{X}^{*}$ denote the unique solution satisfying $\mathbf{AX=b}$. Suppose each agent $v_i$ holds a state vector $\mathbf{\hat x}_i\in\mathbb{R}^d$. Initially, $v_i$ knows $\mathbf{\mathbf{A}}_{i,:}$ and $\mathbf{\mathbf{b}}_{i,:}$. The problem is to devise a local rule for each agent to update its state $\mathbf{\hat x}_i$ leveraging the local communication with agents in $\mathcal{N}_i$ so that $\mathbf{\hat x}_i(t)$ converges to $\mathbf{x}^{*}_i$ within finite $t$. % to  satisfy its own constraint $\sum_{j=1}^n A_{ij}\mathbf{x}_i = \mathbf{b}_i$

%For a system of linear equations, without loss of generality, say $\mathbf{AX=b}$, $\mathbf{A}\in\mathbb{R}^{n\times n}$, $\mathbf{X}\in\mathbb{R}^{n\times d}$, and $\mathbf{b}\in\mathbb{R}^{n\times d}$. $n$ is the number of variables. $d$ is the dimension of a vector, $d\in\{2,3\}$. Let $\mathbf{A}_{i} = [A_{i1} \cdots A_{in}]$ be the $i$th row of $\mathbf{A}$, $i\in\{1\cdots n\}$. 
\subsection{Related Work}
The basic idea of the iterative methods is as follows. Given an initialization $\mathbf{\hat x}(0)$, generate an iteration sequence  $\{\mathbf{\hat x}(t)\}_{t=0}^{\infty}$ in a certain manner, so that:

\begin{equation}
	\label{equ:iteration_sequence}
	\lim\limits_{t \to\infty}\mathbf{\hat x}(t) = \mathbf{x}^{*}\triangleq\mathbf{A}^{-1}\mathbf{b}.
\end{equation}

Generally, the state update of an iterative method can be represented as:
\begin{equation}
	\mathbf{\hat x}(t+1) = \phi_k(\mathbf{\hat x}(t),\mathbf{\hat x}(t-1)\cdots,\mathbf{\hat x}(0),\mathbf{A},\mathbf{b}), 
\end{equation}
where $\mathbf{\hat x}(0)=\phi_0(\mathbf{A},\mathbf{b)}$ or $\mathbf{\hat x}(0)$ is selected manually. $\phi_k$ is called the \emph{iteration function}. The specific design of iterative functions are based on the \emph{matrix splitting}.
\begin{definition}[Matrix Splitting]
	Suppose a  non-singular matrix  $\mathbf{A}\in\mathbb{R}^{n\times n}$, a split of matrix $\mathbf{A}$ is defined  as $\mathbf{A} = \mathbf{M} - \mathbf{N}$, where $\mathbf{M}$ is also non-singular.
\end{definition}

Consider a general linear system:
\begin{equation}
	\label{equ:linear_system}
	\mathbf{AX=b},
\end{equation}
where $\mathbf{A}$ is non-singular. (\ref{equ:linear_system}) can be transformed as follows through matrix splitting $\mathbf{A} =\mathbf{M} - \mathbf{N}$.
\begin{equation}
	\label{equ:linear_split}
	\mathbf{M{X}=N{X}+b}. 
\end{equation}
Then, we can construct an iteration function as:
\begin{equation}
	\label{equ:iteration_function}
	\mathbf{M\hat{X}}(t+1)=\mathbf{N\hat{X}}(t)+\mathbf{b}, 
\end{equation}
which is equivalent to:
\begin{equation}
	\label{equ:iteration_function_transform}
	\mathbf{\hat{X}}(t+1)=\mathbf{M}^{-1}\mathbf{N\hat{X}}(t)+\mathbf{M}^{-1}\mathbf{b} \triangleq \mathbf{G\hat{X}}(t) + \mathbf{g}.
\end{equation}
$\mathbf{G}=\mathbf{M}^{-1}\mathbf{b}$ is called the \emph{iteration matrix}.
It is straightforward that different iterative methods can be constructed by varying $\mathbf{M}$.

In Jacobi iteration, $\mathbf{A}$ is splitted as:

\begin{equation}
	\mathbf{A=D-L-U}.
\end{equation}
$\mathbf{D}$ is the diagonal component of $\mathbf{A}$. $-\mathbf{L}$ and $-\mathbf{U}$ are the upper triangle component and the lower triangle component of $\mathbf{A}$, respectively. Then, the Jacobi iteration is  specified as:

\begin{equation}
	\mathbf{\hat{X}}(t+1) = \mathbf{D}^{-1}(\mathbf{L}+\mathbf{U})\mathbf{\hat{X}}(t) + \mathbf{D}^{-1}\mathbf{b}.
\end{equation}

From the local behavior of an individual agent $v_i$, the state update is:

\begin{equation}
	\mathbf{\hat x}_i(t+1) = \frac{1}{A_{ii}}\left(\mathbf{b}_i-\sum\limits_{j=1,j\neq i}^n A_{ij}\mathbf{\hat x}_j(t)\right)
\end{equation}

In Gauss-Seidel iteration, the iteration function is designed as:

\begin{equation}
	\mathbf{\hat X}(t+1) = (\mathbf{D-L})^{-1}\mathbf{U}\mathbf{\hat X}(t) + (\mathbf{D-L})^{-1}\mathbf{b}.
\end{equation}

The local update of agent $v_i$'s state is:
\begin{equation}
	\mathbf{\hat x}_i(t+1) = \frac{1}{A_{ii}}\left(\mathbf{b}_i-\sum\limits_{j=1}^{i-1} A_{ij}\mathbf{\hat x}_j(t+1) - \sum\limits_{j=i+1}^{n}A_{ij}\mathbf{\hat x}_j(t) \right)
\end{equation}

Another brief iteration method is the Richardson iteration when $\mathbf{A}$ is symmetric positive definite. The iteration function is:

\begin{equation}
	\mathbf{\hat X}(t+1) = \mathbf{\hat X}(t)+ \omega(\mathbf{b}-\mathbf{A}\mathbf{\hat X}(t)).
\end{equation}

The state of agent $v_i$ is updated as:

\begin{equation}
	\mathbf{\hat x}_i(t+1) = \mathbf{\hat x}_i(t)+ \omega\left(\mathbf{b}_i- \sum\limits_{j=1}^{n}\mathbf{A}_{ij}\mathbf{\hat x}_j(t)\right), 
\end{equation}
where $\omega$ is  a non-negative scalar and is suggested to be $\frac{2}{\lambda_{max}+\lambda_{min}}$.	Similar methods include the Successive Over-Relaxation (SOR) iteration, the Symmetric SOR (SSOR) iteration, the Accelerated OR (AOR) iteration, the Symmetric AOR (SAOR) etc \cite{barrett1994templates}. 

\subsection{Convergence and Convergence Rate}
Since the aforementioned methods are iterative, a  crucial issue is to guarantee iteration convergence. For a general iteration function as in (\ref{equ:iteration_function_transform}), its convergence is guaranteed by the following theorem.

%The crucial issue is to guarantee the convergence of iterative approaches.
% for solving (\ref{equ:bll_problem}).
\begin{theorem}
	\label{theorem:correct_convergence}
	The iterates formulated by $\mathbf{\hat X}(t+1)=\mathbf{G}\mathbf{\hat X}(t)+\mathbf{b}$ converges for any $\mathbf{\hat X}(0)$, if and only if $\rho(\mathbf{G})<1$\cite{saad2003iterative}.
\end{theorem} 
$\rho(\mathbf{G})$ is the spectral radius of the iteration matrix $\mathbf{G}$. See Theorem 4.1 of \cite{saad2003iterative} for the proof. 

Apart from knowing when the iteration converges, it is also desirable to explore how fast it converges. Saad \cite{saad2003iterative} presented that the convergence rate $\tau$ is the natural logarithm of the inverse of the spectral radius:

\begin{equation}
	\tau=\ln\frac{1}{\rho(\mathbf{G})} = -\ln\rho(\mathbf{G}).
\end{equation}

It can be seen that these iterative methods converge slowly when $\mathbf{G}$ has unstable eigenvalues.
%Currently, it is realized in two ways. 1) DILOC and its variants introduce new constraints such that $\mathbf{C}$ in (\ref{equ:bll_problem}) is a \emph{substochastic} matrix whose elements are nonnegative and the sum of each row is less than one. $\rho$ of a substochastic matrix has been proved to be less than one\cite{DILOC}. By making $\rho(\bf C) <1$, DILOC guarantees convergence using the iteration in (\ref{equ:linear_iteration_diloc}) to calculate $\mathbf{P}_{\mathcal{S}}$.

\section{DCG: A General Distributed Conjugate Gradient Implementation}
\label{sect:dcg}
The most attractive feature of CG is that it converges within a fixed round of iterations. But CG is essentially a centralized gradient-based solver for linear equations. See Section 6 of \cite{saad2003iterative} for the original CG algorithm. 
Although parallel or distributed CG algorithms have been reported for allocating computational loads in cloud computing\cite{ismail2013implementation} and estimating spectrum in sensor networks\cite{xu2016distributed}, the general distributed CG has not been well touched. 
%In the following, we firstly give a general distributed CG implementation and then apply it to solve the linear localization problem in (\ref{equ:gbll_problem_positive}). 		
%For a general linear system $\mathbf{\mathbf{A} X}=\mathbf{b}$, let $\mathbf{\mathbf{A}}$ = $[\mathbf{\mathbf{A}}_{(1,:)};\cdots;\mathbf{\mathbf{A}}_{(n,:)}]$$\in$$\mathbb{R}^{n\times n}$, $\mathbf{b}$ = $[b_1,\cdots,b_n]^T$$\in$$\mathbb{R}^{n\times 1}$ be the known coefficients and $\mathbf{X}$ = $[X_1,\cdots,X_n]^T$$\in$$\mathbb{R}^{n\times 1}$ be the unknown vector to be calculated.% Initially, $v_i$ knows $\mathbf{\mathbf{A}}_{(i,:)}$ (the $i$th row of $\mathbf{\mathbf{A}}$) and $b_i$.% The DCG problem is to calculate $X_i$ for $v_i$ by neighborhood message passing. The value of $X_i$ at iteration $t$ is denoted by $\widehat{X}_i(t)$, called a state of $v_i$.		

The key difficulty in implementing Distributed CG (DCG) is that the state $\widehat{X}_i$ is updated based on several vectors, however, $v_i$ only knows the $i$th element of these vectors. Specifically, 
%$\widehat{X}_i(t)$ is updated by $\mathbf{\hat{X}}_i$
$\widehat{X}_i(t)=\widehat{X}_i(t-1)+\alpha_i(t){d}_i(t)$. 
$\alpha_i(t)$ and $d_i(t)$ are calculated using the direction vector $\mathbf{d}(t)$ and the residual vector $\mathbf{r}(t)$. But $v_i$ only knows $d_i$ and $r_i$. 
\subsection{Vector Synchronization}

To supplement the necessary information, we design the \emph{Synchronize\_Vector}($z_i$) protocol, through which $v_i$ can gather any complete vector $\mathbf{z}$ by constantly exchanging respective elements with neighbors (Line~\ref{dcg:syn_ini}-\ref{dcg:syn_ret} of Algorithm~\ref{alg:DCG}). $H$ is the largest number of hops throughout the network. So the synchronization can be finished in $H$ rounds. $H$ can be set to $n$ if it is not given at network deployment. The operation `\emph{merge}' (Line~\ref{dcg:syn_merge}) means selecting all non-zero elements of the input vectors and stacking them as a new vector while maintaining their original indexes.

\subsection{Distributed Conjugate Gradient}
\label{sect:general_DCG}
DCG is detailed as \textbf{Function DCG} (Line~\ref{dcg:ini}-\ref{dcg:ret} of Algorithm~\ref{alg:DCG}). At initialization (Line~\ref{dcg:ini}), the state $\widehat{X}_i(0)$ is set to 0. The direction vector component ${d}_i$ and the residual vector component ${r}_i$ are set to 0 and $-{b}_i$, respectively. At each iteration $t$, the local behavior of a node $v_i$ is as follows:
\begin{itemize}
	\item \textbf{Update Residual} (Line \ref{dcg:residual}). 
	The residual ${r}_i$ is updated by	${r}_i(t)=(\mathbf{\mathbf{A}\widehat{X}})_i(t-1)-{b}_i$, which is realized  by ${r}_i(t) = -b_i+\sum\nolimits_{j\in\mathcal{N}_i}\mathbf{A}_{ij}\widehat{X}_j(t-1)$. $\mathbf{\mathbf{A}}_{(i,:)}$ and $b_i$ are known to $v_i$ at the begining. $\widehat{X}_j(t-1)$ is the state of neighbor $v_j$ obtained by local communication. 
	\item \textbf{Check Residual} (Line \ref{dcg:terminate}). %\textbf{Termination Judgement}
	CG theoretically completes after $n$ iterations \cite{Wolfgang2016}. However, due to the accumulated floating point rounding off errors, the residual and the direction gradually lose accuracy. Thus, DCG terminates by checking whether $\mathbf{r}(t)^T\mathbf{r}(t)<\varepsilon$. The threshold $\varepsilon$ is an empirical value based on accuracy requirement. 	$\mathbf{r}(t)$ is obtained by: $\mathbf{r}(t)=$  \emph{Synchronize\_Vector}($r_i(t)$). After knowing the synchronized $\mathbf{r}$, the squared residual is calculated as $\mathbf{r}(t)^T\mathbf{r}(t)=\sum_{i=1}^{n}r_i(t)^2$. 
	%Then, $\mathbf{r}(t)^T\mathbf{r}(t)$ can be calculated for judgement.
	\item \textbf{Update Direction} (Line \ref{dcg:direction}). 
	$d_i$ needs another vector $\mathbf{r}(t-1)$ to be updated. It is  obtained by \emph{Synchronize\_Vector}($r_i(t-1)$). Then, $\mathbf{r}(t-1)^T\mathbf{r}(t-1)=\sum_{i=1}^{n}r_i(t-1)^2$ and $d_i$ can be calculated locally.
	\item \textbf{Update Step Size} (Line \ref{dcg:step}). 
	%To calculate $\alpha_i$, represent its numerator and denominator as $numer=\mathbf{d}(t)^T\mathbf{r}(t)$ and $denom=\mathbf{d}(t)^{T}\mathbf{\mathbf{A}}\mathbf{d}(t)$, respectively. 
	Represent the numerator and denominator of $\alpha_i$ as  $numer=\mathbf{d}(t)^T\mathbf{r}(t)$ and $denom=\mathbf{d}(t)^{T}\mathbf{\mathbf{A}}\mathbf{d}(t)$, respectively. 
	The direction vector $\mathbf{d}(t)$ is obtained by \emph{Synchronize\_Vector}($d_i(t)$). So the numerator can be known as $numer=\sum_{i=1}^{n}d_i(t)r_i(t)$. 
	
	An intermediate vector $\mathbf{T}(t)=[T_1(t)\cdots T_n(t)]^T$ is introduced to calculate $denom$. By communicating with neighbors, $v_i$ calculates ${T}_i(t)= \sum\nolimits_{j\in\mathcal{N}_i}\mathbf{A}_{ij}{d}_{j}(t)$. The complete $\mathbf{T}(t)$ is obtained by \emph{Synchronize\_Vector}($T_i(t)$). So $denom=\sum\nolimits_{i=1}^{n}d_i(t)T_i(t)$. Then $\alpha_i=-numer/denom$.
	%To calculate $\alpha_i$, the direction vector is obtained by: $\mathbf{d}(t)=$ \emph{Synchronize\_Vector}($d_i(t)$). Then, the numerator $\mathbf{d}(t)^T\mathbf{r}(t)$ is obtained; The denominator $\mathbf{d}(t)^{T}\mathbf{\mathbf{A}}\mathbf{d}(t)$ is calculated by introducing an intermediate vector $\mathbf{T}$. $v_i$ locally calculates ${T}_i(t)= \sum\nolimits_{j\in\mathcal{N}_i}A_{ij}{d}_{j}(t)$. The complete $\mathbf{T}$ is obtained by: $\mathbf{T}(t)=$ \emph{Synchronize\_Vector}($T_i(t)$). The denominator is also known. Then, $\alpha_i$ is known.
	\item \textbf{Update State}  (Line \ref{dcg:estimation}). 
	The state $\widehat{X}_i(t)$ updates a step $\alpha_i(t){d}_i(t)$ along its last state.	
\end{itemize}
Overall, all of the above DCG procedures are completely distributed and implemented through neighborhood message passing.

	\begin{algorithm}[h]
		\caption{\textbf{Distributed Conjugate Gradient (DCG) of $v_i$}}
		\label{alg:DCG} 
		\KwIn{$\mathbf{\Omega}_{(i,:)}$; $b_i$\;}
		\KwOut{$\mathbf{\hat x}_i$\;}
		
		%\textbf{Function DCG}($\mathbf{\Omega}_{(i,:)}$,$b_i$)\\
		{
			\lnl{dcg:ini} $\widehat{X}_i(0)\gets 0$; ${d}_i(0)\gets 0$;
			${r}_i(0)\gets -b_i$ \tcp{\small{Initialize}} 
			\tcc{\small{Update $\widehat{X}_i$ as in Section~\ref{sect:general_DCG}}}			
			\nl \For{iterations $t\in\{1,\cdots,t_{max}\}$}
			{
				\lnl{dcg:residual} ${r}_i(t)=(\mathbf{\Omega \widehat{X}}(t-1))_i-{b}_i$; \tcp{\small{residual}}
				\lnl{dcg:terminate} \If{$\mathbf{r}(t)^T\mathbf{r}(t)<\varepsilon$}
				{
					\nl\textbf{break} \tcp{\small{meets accuracy requirement}}
				}
				\lnl{dcg:direction} ${d}_i(t)$=$-{r}_i(t)$+$\frac{\mathbf{r}(t)^{T}\mathbf{r}(t)}{\mathbf{r}(t-1)^{T}\mathbf{r}(t-1)}$${d}_i(t-1)$ %\tcp{\scriptsize{direction}}  
				$//$\texttt{\small{direction}}\\
				\lnl{dcg:step} $\alpha_i(t)=-\frac{\mathbf{d}(t)^{T}\mathbf{r}(t)}{\mathbf{d}(t)^{T}\mathbf{\Omega}\mathbf{d}(t)}$; \tcp{\small{step size}}   
				\lnl{dcg:estimation} $\widehat{X}_i(t)=\widehat{X}_i(t-1)+\alpha_i(t){d}_i(t)$; \tcp{\small{state}}
			}
		}
		\lnl{dcg:ret}\KwRet{$\widehat{X}_i(t)$.}
		\vspace{5pt}\\
		\textbf{Function Synchronize$\_$Vector}($z_i$)\\
		{
			\lnl{dcg:syn_ini} initialize $\mathbf{z}_i(0)$ $\gets$ $[\boldsymbol{0}_{(i-1)\times 1}~z_i~\boldsymbol{0}_{(n-i)\times 1}]$\;
			\nl \For {iterations $k\in\{0,\cdots,H\}$}
			{
				\nl exchange $\mathbf{z}_i(k)$ with $\mathcal{N}_i$\;
				\nl \For{each $v_j\in\mathcal{N}_i$}
				{
					\lnl{dcg:syn_merge} $\mathbf{z}_i(k+1)$  $\gets$ \textbf{merge}($\mathbf{z}_i(k)$, $\mathbf{z}_{j}(k)$)\;
				}
			}
			\lnl{dcg:syn_ret}\KwRet{$\mathbf{z}_i$}.
		}
		
	\end{algorithm}

\subsection{Analysis of DCG}

In each round, four vectors are obtained by synchronization, so it converges within $4Hn$ rounds when $\mathbf{A}$ is non-singular. In actual applications, there may exist measurement noise so that $\mathbf{A}$ and $\mathbf{b}$ are influenced and the constructed linear system may be unsolvable\cite{Xie2020}. Let $\overline{\mathbf{A}}$ and $\overline{\mathbf{b}}$ denote the noisy matrices as:
\begin{equation}
	\label{equ:noise}
	\overline{\mathbf{A}} = \mathbf{A}+\Delta\mathbf{A}, \overline{\mathbf{b}} = \mathbf{b}+\Delta\mathbf{b},
\end{equation}
where $\Delta\mathbf{A}$ and $\Delta\mathbf{b}$ are error matrices implying the noise. If the noisy matrix $\overline{\mathbf{A}}$ is still non-singular, DCG still converges to the neighborhood of $\mathbf{X}^*$.
\begin{lemma}
	For a noisy linear system $\overline{\mathbf{A}}\mathbf{X}=\mathbf{b}$, DCG converges if the error matrix satisfies:
	\begin{equation}
		\label{equ:noise_converge}
		||\Delta\mathbf{A}||<\lambda_{min}(\mathbf{A}).
	\end{equation}
\end{lemma}
\begin{proof}
	From (\ref{equ:noise}), we can obtain:
	\begin{equation}
		\overline{\mathbf{A}} = \mathbf{A}(I+\mathbf{A}^{-1}\Delta\mathbf{A}).
	\end{equation}
	Since $||\mathbf{A}^{-1}\Delta\mathbf{A}||\leq||\mathbf{A}^{-1}||~||\Delta\mathbf{A}||=\frac{||\Delta\mathbf{A}||}{\lambda_{min}(\mathbf{A})}$, the matrix $\overline{\mathbf{A}}$ must be non-singular if (\ref{equ:noise_converge}) holds. Thus DCG can converge to the neighborhood of $\mathbf{X}^*$ shown as: $\mathbf{\overline X}^*=\overline{\mathbf{A}}^{-1}\overline{\mathbf{b}}$.
\end{proof}
%\section{DCG-Loc: DCG for Distributed Network Localization}
\section{Applications of DCG}
\label{sect:applications}
In this section, we investigate applying DCG to two actual scenarios.

\subsection{The Least Square Problem}
Linear equations arising from the era of engineering are usually over-determined. A typical scenario is distributed parameter estimation, where the observation equations are as:

\begin{equation}
	A_i x_i = b_i+\delta_i,
\end{equation}
where $x_i$ is the desired parameter to be calculated and $\delta_i$ is the component implying the measurement noise of $b_i$. Due to the measurement noise, such linear equations usually do not have a solution that exactly meets the constraints. This problem can be formulated as follows. Suppose there is no solution to the linear system $\mathbf{AX=b}$ and $\mathbf{AA}^{T}$ is non-singular. Each agent $v_i$ only knows the $i$th row of $\mathbf{A}$ and $\mathbf{b}$, which are denoted by  $\mathbf{A}_{i,:}$ and $\mathbf{b}_{i,:}$, respectively. Design a distributed rule for each agent to update its state $\mathbf{x}_i$ so that $\mathbf{x}_i(t)$ converges to the unique solution of:
\begin{equation}
	\label{equ:least_square}
	\mathbf{A}^{T}\mathbf{AX}=\mathbf{A}^{T}\mathbf{b}.
\end{equation} 

Let $\boldsymbol{\Omega} =\mathbf{A}^{T}\mathbf{A}$ and $\boldsymbol{\beta} =\mathbf{A}^{T}\mathbf{b}$. 
To solve the least square problem as in (\ref{equ:least_square}), each agent should be aware of the $i$th row of $\boldsymbol{\Omega}$ and $\boldsymbol{\beta}$ using $\mathbf{A}_{i,:}$ and $\mathbf{b}_{i,:}$. 

Initially, $v_i$ maintains $\mathbf{A}_{i,:}$.%, so it knows the nonzero elements of the $i$th row $\mathbf{M}_{i,:}$ by ${M}_{sub(i,j)}=-A_{sub(i,j)}$ if $v_j\in\mathcal{N}_i^*$. %, ${M}_{sub(i,j)}=0$, otherwise. %if $v_j\notin\mathcal{N}_i^*$. %, and ${M}_{sub(i,i)}=1$
%Although $\mathbf{M}_{sub}$ is not symmetric, $\mathcal{G_A}$ is an undirected graph so that the nonzero elements of $\mathbf{M}^T_{sub(i,:)}$ are in the same locations as $\mathbf{M}_{sub(i,:)}$. 
Then, $v_i$ transmits ${A}_{i,j}$ to $\mathcal{N}_i$ and receives ${A}_{j,i}$ from $\mathcal{N}_i$, then  it also knows the nonzero elements of the $i$th column $\mathbf{A}_{:,i}$, i.e., $\mathbf{A}^T_{i,:}$. Thus the nonzero elements of $\boldsymbol{\Omega}_{i,:}$ are calculated distributively as:
%$\mathbf{M}_{sub(:,i)}=[{M}_{sub(1,i)},\cdots,{M}_{sub(n,i)}]^{T}$, i.e., $\mathbf{M}^T_{sub(i,:)}$. Thus $\boldsymbol{\mathbf{A}}_{(i,:)}$ can be calculated distributively as:
% \begin{equation}
% 	\label{equ:dcg_ini_omega}
% 	\mathbf{\mathbf{A}}_{(i,:)}=[\mathbf{A}_{i,1}\cdots\mathbf{A}_{i,n}],  \mathbf{A}_{ij}=\sum\nolimits_{k=1}^{n} {M}^T_{sub(i,k)}M_{sub(k,j)}. 
% \end{equation}
\begin{equation}
	\label{equ:ini_omega}
	{\Omega}_{ij}=\sum\nolimits_{k=1}^{n} {A}^T_{ik}A_{kj},\forall v_j\in\mathcal{N}_i. 
\end{equation}

Similarly, $\boldsymbol{\beta}_{i,:}$ can be calculated as:
\begin{equation}
	\label{equ:ini_beta}
	{\beta}_{ij}=\sum\nolimits_{k=1}^{n} {A}^T_{ik}b_{kj},  j\in\{1,\cdots,d\}. 
\end{equation}

Finally, the problem can be solved as $\mathbf{x}_i=DCG(\boldsymbol{\Omega}_{i,:}, \boldsymbol{\beta}_{i,:})$. Therefore $\mathbf{x}$ solves the least squares problem in (\ref{equ:least_square}).

\subsection{The Network Localization Problem}
%After PEP and neighbor selection in GBLL, the locations of BLL-localizable agents are formulated as (\ref{equ:gbll_problem}). 
The geographical locations of nodes are fundamental information for many multi-agent applications\cite{Nguyen2020,Sun2018,Wang2018}. Network localization techniques are usually adopted for calculating node locations in infrastructure-less scenarios\cite{Ping2020,PingHGO}, which are formulated as follows.
For a network of $m+n$ agents in $\mathbb{R}^d$, let $\mathcal{V}=\mathcal{A}\cup\mathcal{F}$ denote the entire node set.   Nodes in $\mathcal{A}=\{v_1,\cdots,v_m\}$ are called \emph{anchor agents}, whose locations $\mathbf{P}_{\mathcal{A}} = \{\mathbf{p}_1,\cdots,\mathbf{p}_m\}$ are known.  Nodes in $\mathcal{F}= \{v_{m+1},\cdots,v_{m+n}\}$ are called \emph{free agents}, whose locations $\mathbf{P}_{\mathcal{F}} = \{\mathbf{p}_{m+1},\cdots,\mathbf{p}_{m+n}\}$ are unknown. Each agent $v_i$ can only sense the relative distance $d_{ij}$ between $v_i$ and any neighbor $v_j\in\mathcal{N}_i$. Each agent can exchange its estimated location  $\hat{\mathbf{p}}_i$ and distance measurements with neighbors. The network localization problem is to design a distributed protocol for each agent to update its location $\hat{\mathbf{p}}_i$ so that it converges to $\mathbf{p}_i$. 

To solve the network localization problem, we transform it to a linear system. First, each location $\mathbf{p}_i$ is represented as a linear combination of locations of neighbors:
\begin{equation}
	\label{equ:cal_bary}
	\mathbf{p}_i = \sum_{v_j\in\mathcal{N}_i} a_{ij}\mathbf{p}_j,
\end{equation}  
where $a_{ij}$ are called \emph{barycentric coordinates}. The calculation of barycentric coordinates involve only local distance measurements and the specific process is introduced by Diao et al. \cite{ECHO} in $\mathbb{R}^2$ and by Han et al. \cite{Han2017} in $\mathbb{R}^3$. Then after calculating the barycentric coordinates for each node, the agent locations can form a linear system:
\begin{equation}
	\label{equ:linear_representation}
	\left[ {\begin{array}{*{20}{c}}
			{{{\mathbf{P}}_{\mathcal{A}}}} \\ 
			{{{{\mathbf{ P}}}_{\mathcal{F}}}} 
	\end{array}} \right] = \left[ {\begin{array}{*{20}{c}}
			{\mathbf{I}}&{\mathbf{0}} \\ 
			{\mathbf{B}}&{\mathbf{C}} 
	\end{array}} \right]\left[ {\begin{array}{*{20}{c}}
			{{{\mathbf{P}}_{\mathcal{A}}}} \\ 
			{{{{\mathbf{P}}}_{\mathcal{F}}}} 
	\end{array}} \right].
\end{equation}
$\mathbf{A}=\left[ {\begin{array}{*{20}{c}}
		{\mathbf{I}}&{\mathbf{0}} \\ 
		{\mathbf{B}}&{\mathbf{C}} 
\end{array}} \right] \in\mathbb{R}^{(m+n)\times(m+n)}$ is constructed with barycentric coordinates, i.e., the $i$th row of $\mathbf{A}$  is the barycentric coordinate of $v_i$ w.r.t. its neighbors $\mathcal{N}_i$. Then, the localization problem can be transformed into  solving the following linear system: 
\begin{equation}
	\label{equ:bll_problem}
	(\mathbf{I}-\mathbf{C})\mathbf{P}_{\mathcal{F}} = \mathbf{B}\mathbf{P}_{\mathcal{A}}.
\end{equation}
Writing $\mathbf{I}-\mathbf{C}$ as $\mathbf{M}$, (\ref{equ:bll_problem})  can be reformulated as:
\begin{equation}
	\label{equ:gbll_problem_reformulate}
	%\mathbf{M}_{sub}^{T}\mathbf{M}_{sub}\mathbf{P}_{\mathcal{S}^{*}} = \mathbf{M}_{sub}^{T}\mathbf{B}_{sub}\mathbf{P}_{\mathcal{A}}
	\mathbf{M}\mathbf{P}_{\mathcal{F}} = \mathbf{B}\mathbf{P}_{\mathcal{A}}.
\end{equation}
%Although Richardson-iteration can be used to solve $\mathbf{P}_{\mathcal{S}^{*}}$ from (\ref{equ:gbll_problem_reformulate}) \cite{ECHO}\cite{Han2017}, its convergence rate is slow. For speeding up, we turn to Conjugate Gradient (CG) iteration. CG has two desired properties for this problem\cite{shewchuk1994introduction}\cite{Wolfgang2016}: 1) CG converges to the exact solution after a finite number of iterations, which is not larger than the size of the system matrix; 2) CG is suited for solving linear systems with large and sparse system matrices. 
Considering that CG is used when the system matrix is positive definite \cite{Wolfgang2016}. Thus, we multiply $\mathbf{M}^{T}$ to both sides of (\ref{equ:gbll_problem_reformulate}):
%Usually, CG is effective for systems of the symmetric and positive-definite form\cite{\label{shewchuk1994introduction}}. 
%To meet this restriction, we multiply $\mathbf{M}_{sub}^{T}$ to both sides of (\ref{equ:gbll_problem_reformulate}):
\begin{equation}
	\label{equ:gbll_problem_positive}
	\mathbf{M}^{T}\mathbf{M}\mathbf{P}_{\mathcal{F}} = \mathbf{M}^{T}\mathbf{B}\mathbf{P}_{\mathcal{A}}.
\end{equation}
Then, $\mathbf{M}^{T}\mathbf{M}$ is positive definite. 		
To apply DCG to distributed localization, the  localization model in (\ref{equ:gbll_problem_positive}) is reformulated to $\boldsymbol{\Omega}\mathbf{P}_{\mathcal{F}}=\boldsymbol{\beta}$, where $\boldsymbol{\Omega}$ = $\mathbf{M}^{T}\mathbf{M}\in\mathbb{R}^{n\times n}$ and $\boldsymbol{\beta}$ = $\mathbf{M}^{T}\mathbf{B}\mathbf{P}_{\mathcal{A}}\in\mathbb{R}^{n\times d}$. 

Algorithm~\ref{alg:DCG_Loc} shows the routine of DCG-Loc. For initialization, $v_i$ needs to know the $i$th row of $\boldsymbol{\Omega}$ and $\boldsymbol{\beta}$ to invoke the general DCG. 
%In DCG-Loc, each $v_i$ locally calculates the $i$th row of $\boldsymbol{\mathbf{A}}$ and $\boldsymbol{\beta}$ to localize itself leveraging DCG. Algorithm~\ref{alg:DCG} gives the overall routine. 
After constructing the local linear model by (\ref{equ:cal_bary}), $v_i$ maintains $\mathbf{A}_{i,:}$, so it knows the nonzero elements of the $i$th row $\mathbf{M}_{i,:}$ by ${M}_{i,j}=-A_{i,j}$ if $v_j\in\mathcal{N}_i$. %, ${M}_{sub(i,j)}=0$, otherwise. %if $v_j\notin\mathcal{N}_i^*$. %, and ${M}_{sub(i,i)}=1$
%Although $\mathbf{M}_{sub}$ is not symmetric, $\mathcal{G_A}$ is an undirected graph so that the nonzero elements of $\mathbf{M}^T_{sub(i,:)}$ are in the same locations as $\mathbf{M}_{sub(i,:)}$. 
Then, $v_i$ transmits ${M}_{ij}$ to $\mathcal{N}_i$ and receives ${M}_{ji}$ from $\mathcal{N}_i$ (Line~\ref{dcg:send_M}-\ref{dcg:receive_M}), so it also knows the nonzero elements of the $i$th column $\mathbf{M}_{:,i}$, i.e., $\mathbf{M}^T_{i,:}$. Thus the nonzero elements of $\boldsymbol{\Omega}_{i,:}$ are calculated distributively as:
%$\mathbf{M}_{sub(:,i)}=[{M}_{sub(1,i)},\cdots,{M}_{sub(n,i)}]^{T}$, i.e., $\mathbf{M}^T_{sub(i,:)}$. Thus $\boldsymbol{\mathbf{A}}_{(i,:)}$ can be calculated distributively as:
% \begin{equation}
% 	\label{equ:dcg_ini_omega}
% 	\mathbf{\mathbf{A}}_{(i,:)}=[\mathbf{A}_{i,1}\cdots\mathbf{A}_{i,n}],  \mathbf{A}_{ij}=\sum\nolimits_{k=1}^{n} {M}^T_{sub(i,k)}M_{sub(k,j)}. 
% \end{equation}
\begin{equation}
	\label{equ:dcg_ini_omega}
	\boldsymbol{\Omega}_{ij}=\sum\nolimits_{k=1}^{n} {M}^T_{ik}M_{kj},\forall v_j\in\mathcal{N}_i. 
\end{equation}
To calculate $\boldsymbol{\beta}_{(i,:)}$, an intermediate vector $\boldsymbol{\mu}\in\mathbb{R}^{n\times d}$ implying $\mathbf{B}\mathbf{P}_{\mathcal{A}}$ is introduced. $v_i$ locally calculates the $i$th row of $\boldsymbol{\mu}$: 
\begin{equation}
	\label{equ:dcg_ini_mu}
	\boldsymbol{\mu}_{i,:}=\sum\nolimits_{v_a\in \mathcal{N}_i\cap\mathcal{A}} A_{ia}\mathbf{p}_a,
\end{equation}
where $\mathcal{N}_i\cap\mathcal{A}$ represents neighboring anchors. %An anchor $v_a$'s location $\mathbf{p}_a=[{p}_a^x,~{p}_a^y]$ is known to $v_i$ if $v_a\in\mathcal{N}^*_i$. 
$\boldsymbol{\mu}_{i,:}$ = $[0,0]$ if no anchor is found in $\mathcal{N}_i$. Then, $v_i$ sends $\boldsymbol{\mu}_{i,:}$ to $\mathcal{N}_i$ and receives $\boldsymbol{\mu}_{j,:}$  of each $v_j\in\mathcal{N}_i^*$ (Line~\ref{dcg:send_M}-\ref{dcg:receive_M}). Thus, $\boldsymbol{\beta}_{i,:}$ is calculated as:

\begin{equation}
	\label{equ:dcg_ini_beta}
	\boldsymbol{\beta}_{i,:} = \sum\nolimits_{v_j\in \mathcal{N}_i\cap\mathcal{F}}M^T_{i,j}\boldsymbol{\mu}^{}_{j,:},
\end{equation} %j\in\mathcal{N}^*_i\cap\mathcal{S}
where $\mathcal{N}_i\cap\mathcal{F}$ means the non-anchor barycentric neighbors. 		
For convenience, each location $\mathbf{\hat{p}}_i\in\mathbb{R}^{d\times 1}$ is decomposed to $[\hat{p}^{1}_i,\cdots,\hat{p}^{d}_i]$.  $\boldsymbol{\beta}_{(i,:)}$ is decomposed to $[\beta_i^{1},\cdots,\beta_i^{d}]$. The element $\hat{p}^j_i$ is calculated by \textbf{DCG}($\boldsymbol{\mathbf{A}}_{i,:}$, $\beta_i^j$) (Line~\ref{dcg:ret_loc}). Therefore, from the GBLL model in (\ref{equ:gbll_problem_positive}), $\mathbf{\hat{p}}_i$ is calculated leveraging DCG-Loc, where communications only involve message passing with neighbors.

%The two elements $\hat{p}^x_i$ and $\hat{p}^y_i$ are calculated by \textbf{DCG}($\boldsymbol{\mathbf{A}}_{(i,:)}$, $\beta_i^x$) and \textbf{DCG}($\boldsymbol{\mathbf{A}}_{(i,:)}$, $\beta_i^y$), respectively (Line~\ref{dcg:ret_loc}). Therefore, from the GBLL model in (\ref{equ:gbll_problem}), $\mathbf{\hat{p}}_i$ is calculated leveraging DCG-Loc, where communications only involve message passing with neighbors.

	\begin{algorithm}[h]
		\caption{\textbf{Distributed Conjugate Gradient Localization (DCG-Loc) of $v_i$}}
		%\caption{\textbf{Distributed Conjugate Gradient Localization (DCG-Loc) of $v_i$}}
		\label{alg:DCG_Loc} 
		\KwIn{neighbors: $\mathcal{N}_i$;  barycentric coordinates: $\mathbf{A}_{i,:}$\;} %max hop count: $H$\; max iteration count: $t_{max}$; convergence threshold: $\varepsilon$
		\KwOut{location: $\mathbf{\hat p}_i$\;}
		%\tcc{\small{Initialize $\boldsymbol{\Omega}$,   $\boldsymbol{\beta}$ as in Section~\ref{sect:DCG_Loc}}}
		%$/*~$\texttt{\small{Initialize $\boldsymbol{\Omega}$,   $\boldsymbol{\beta}$ as in Section~\ref{sect:DCG_Loc}}}$*/$\\
		\lnl{dcg:init_M} $M_{ij}\gets-A_{ij}$; calculate $\boldsymbol{\mu}_{i,:}$ as (\ref{equ:dcg_ini_mu})\; %$v_j\in\mathcal{N}_i^*$
		\lnl{dcg:send_M} transmit $M_{ij}$ and $\boldsymbol{\mu}_{i,:}$ to $v_j\in\mathcal{N}_i$\;
		\lnl{dcg:receive_M} receive $M_{ji}$ and $\boldsymbol{\mu}_{j,:}$ from $v_j\in\mathcal{N}_i$\; %$//$\texttt{\small{$i$th column of $M_{sub}$}}\\
		%\nl calculate $\boldsymbol{\Omega}_{(i,:)}$ and  $\boldsymbol{\beta}_{(i,:)}$ as (\ref{equ:dcg_ini_omega}) and (\ref{equ:dcg_ini_beta}), respectively\;
		\lnl{dcg:init_omega} calculate $\boldsymbol{\Omega}_{i,:}$ as (\ref{equ:dcg_ini_omega}), calculate $\boldsymbol{\beta}_{i,:}$ as (\ref{equ:dcg_ini_beta})\;
		\lnl{dcg:ret_loc}\KwRet $\mathbf{\hat{p}}_i\gets$ [\textbf{DCG}($\boldsymbol{\Omega}_{i,:}$, $\beta_i^1$),$\cdots$, \textbf{DCG}($\boldsymbol{\Omega}_{i,:}$, $\beta_i^d$)].
		
	\end{algorithm}

\section{Evaluation}
\label{sect:evaluation}

In this section, we evaluate the convergence speed between our proposed DCG algorithm and the representative Richardson iteration by counting the iteration rounds. Simulations are conducted using MATLAB R2020b in both $\mathbb{R}^2$ and $\mathbb{R}^3$. A network denoted by $\mathcal{G=\{V,E\}}$ is deployed. The number of anchors is set to $d+1$, which is the minimum number of anchors required to uniquely localize a network in $\mathbb{R}^d$.

Fig. \ref{fig:show_iter_2d} and Fig. \ref{fig:effectiveness_dcg_2d} shows the results in $\mathbb{R}^2$. DCG  and Richardson iteration are adopted to solve the linear
localization problem modeled by the network in Fig. \ref{fig:show_iter_2d}. It is shown that both DCG  and Richardson successfully converge to the ground truth within finite rounds of iterations. However, Fig. \ref{fig:effectiveness_dcg_2d} shows that the Richardson iteration consumes $5\times10^{5}$ rounds while DCG only needs 550 rounds. 		 
From Fig. \ref{fig:show_iter_3d} and Fig. \ref{fig:effectiveness_dcg_3d}, similar results can be obtained in $\mathbb{R}^3$. Overall, DCG-Loc is shown to be faster than Richardson iteration
about 1,000 times.

\begin{figure}
	\centering
	\subfigure[The Richardson Iteration]{
		\begin{minipage}[t]{0.46\linewidth}
			\centering
			\includegraphics[width=1\textwidth]{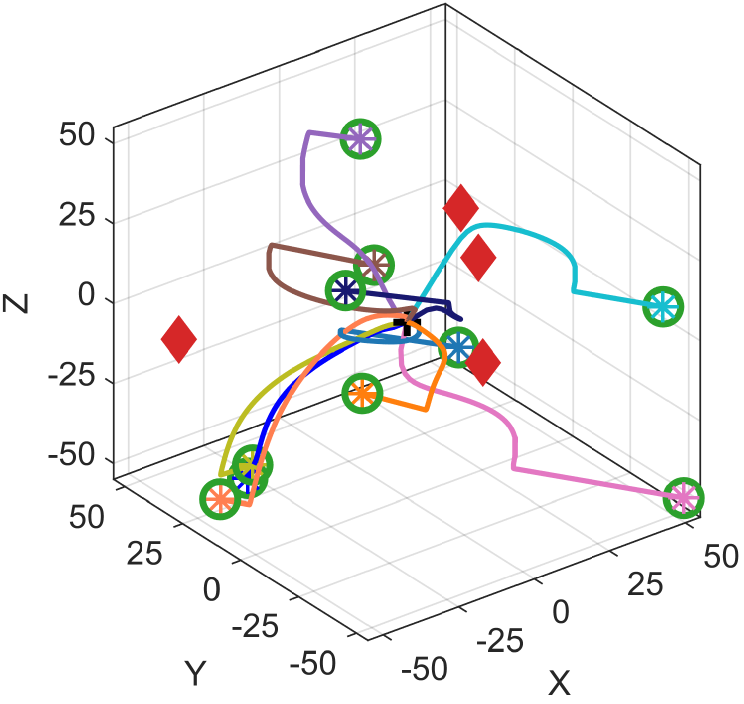}
			\label{fig:trace_richardson_3d}
		\end{minipage}
	}
	\vspace{-0.2cm}
	\subfigure[DCG]{
		\begin{minipage}[t]{0.46\linewidth}
			\centering
			\includegraphics[width=1\textwidth]{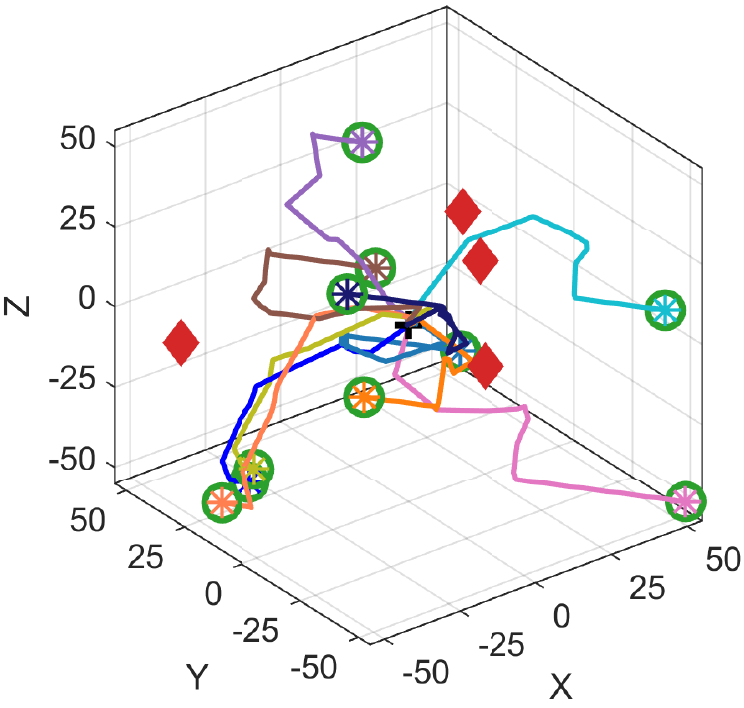}
			\label{fig:trace_dcg_3d}
		\end{minipage}
	}
	\caption{The convergence trails in $\mathbb{R}^3$. }%The marker `+' represents the start a trail. The marker `$*$' and `$\circ$' represent the converged estimation $\mathbf{\hat{x}}_i$ and the ground truth  $\mathbf{x}_i$, respectively. 
	\label{fig:show_iter_3d}
\end{figure}

\begin{figure}		
	\centering	
	\includegraphics[width=.99\linewidth]{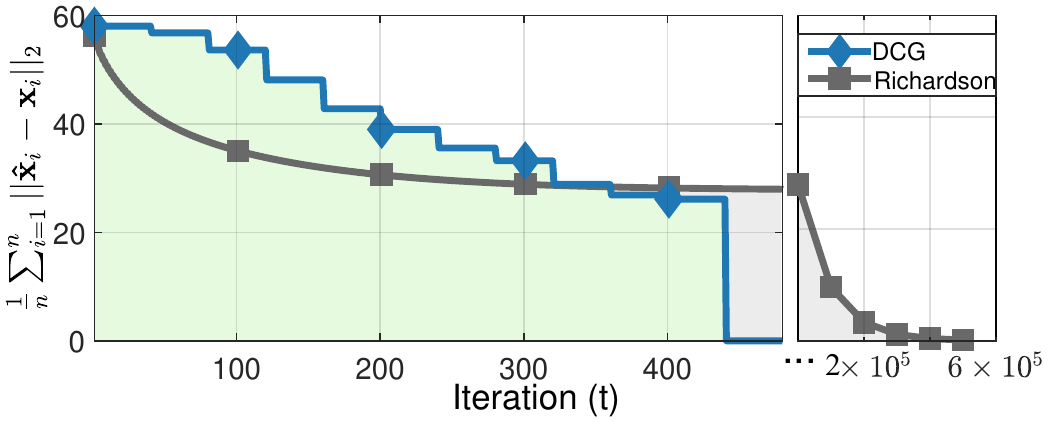}
	\caption{The mean square error w.r.t. iteration rounds.}
	\label{fig:effectiveness_dcg_3d}
\end{figure}

\section{Conclusion}
\label{sect:conclusion}
In this paper, we proposed DCG to enable a network of $n$ agents to solve linear equations like $\mathbf{AX=b}$ in fixed rounds of iterations. The DCG algorithm is presented with property analysis and two applications. Compared with traditional Richardson iteration, DCG shows 3 magnitudes faster convergence speed. In future work, we will consider reducing the communication burden that DCG requires.

\bibliographystyle{unsrt}
\bibliography{main}

\end{document}